\documentclass[journal,twoside,web]{ieeecolor}
\usepackage{lcsys}
\usepackage{cite}
\usepackage{amsmath,amssymb,amsfonts}
\usepackage{algorithmic}
\usepackage{graphicx}
\usepackage{mathtools}
\usepackage{multirow}
\usepackage{booktabs}
\usepackage{tabularx}

\definecolor{purple}{rgb}{0.8,0,0.3}
\definecolor{deepblue}{rgb}{0.2,0,1}

\makeatletter
\let\NAT@parse\undefined
\makeatother
\usepackage{hyperref}
\hypersetup{
    colorlinks,
    linkcolor={deepblue},
    citecolor={purple},
    urlcolor={black}
}

\def\BibTeX{{\rm B\kern-.05em{\sc i\kern-.025em b}\kern-.08em
    T\kern-.1667em\lower.7ex\hbox{E}\kern-.125emX}}
\newcommand{\Qa}{\ensuremath{Q_\alpha}}
\newcommand{\Pa}{\ensuremath{P_\alpha}}
\DeclareMathOperator{\trace}{trace}
\DeclareMathOperator{\tr}{trace}
\newtheorem{theorem}{Theorem}
\newtheorem{lemma}{Lemma}
\newtheorem{remark}{Remark}

\begin{document}
\title{Optimal Control for Discrete-Time Systems under Bounded Disturbances}
\author{Egor Dogadin, Alexey Peregudin, Dmitriy Shirokih
\thanks{Egor Dogadin and Dmitriy Shirokih are with the Faculty of Control Systems and Robotics, ITMO University, Kronverkskiy av. 49, St. Petersburg, 197101, Russia (email: egor.dogadin@icloud.com; sher145@yandex.ru).}
\thanks{Alexey Peregudin is with the Faculty of Control Systems and Robotics, ITMO University, Kronverkskiy av. 49, St. Petersburg, 197101, Russia Russia (e-mail: peregudin@itmo.ru).}
\thanks{This work was supported by the Ministry of Science and Higher Education of Russian Federation according to the research project (Goszadanie) no.2019-0898.}
}

\maketitle

\begin{abstract}
This paper introduces a novel approach to the optimal control of linear discrete-time systems subject to bounded disturbances. 
Our approach is based on the newly established duality between ellipsoidal approximations of reachable and hardly observable sets. We provide exact solutions for state-feedback control and filtering problems, aligning with existing methods while offering improved computational efficiency. Moreover, our main contribution is the optimal solution to the output-feedback control problem for discrete-time systems which was not known before. Numerical simulations demonstrate the superiority of this result over previous sub-optimal ones.
\end{abstract}


\section{Introduction}
\label{sec:introduction}

\IEEEPARstart{O}{ptimal} control of dynamic systems under bounded disturbances remains a central challenge in control theory. While well-established techniques like LQG, $\mathcal{H}_2$, and $\mathcal{H}_\infty$ control can effectively attenuate specific types of disturbances, such as Gaussian white noise or harmonic oscillations, addressing the general optimal control problem in the presence of bounded disturbances remains a complex task.

A significant advancement in this field came with the introduction of the invariant (attractive) ellipsoids method. This method allows for the analysis and controller design for systems subject to general bounded disturbances within the $L_\infty$-class. It provides a geometric interpretation and aims to minimize the size of the invariant set of the closed-loop system, reducing the effect of external disturbances. The foundations of the invariant ellipsoids method for discrete-time systems were initially laid out in \cite{b12-boyd}, \cite{b5-abedor}, and later developed in \cite{b1-nazin, b2-topunov, Khlebnikov2011}. The method is used for state-feedback control \cite{b1-nazin}, output-feedback control \cite{b2-topunov} and observer design \cite{Khlebnikov2011}. Its further development included network systems control \cite{x2}, adaptive control \cite{y1}, and robust control \cite{x6, x7, e5, e6}. Some of these results are summarized in the classical monograph \cite{b3-poznyak}.

Despite the widespread use of the invariant ellipsoids method, it is known to provide only sub-optimal solutions to the output-feedback control problem, as noted in \cite{Khlebnikov2011}. In this paper, we revisit the fundamental findings of the invariant ellipsoids method concerning state-feedback and filtering but take a novel approach to reestablish them. This approach will also enable us to derive the optimal solution for the output-feedback problem. At the core of our methodology lies the introduction of a new concept of hardly observable sets,  which are the dual counterparts to the reachable sets traditionally minimized using the invariant ellipsoids method.

The primary contributions of this paper are as follows:
\begin{itemize}
    \item We establish a duality between the ellipsoidal approximations of reachable and hardly observable sets for discrete-time systems, which leads to a definition of the $\varepsilon$-norm.
    \item For the first time, we present exact equations for determining optimal state-feedback and observer gains in discrete-time systems with respect to the $\varepsilon$-norm. These equations yield the same results as the optimization procedures proposed in \cite{b1-nazin, b2-topunov, Khlebnikov2011}, but are computationally more efficient due to a reduced number of variables.
    \item We introduce an optimal solution to the output-feedback control problem for discrete-time systems under bounded disturbances. This solution shows a significant improvement over prior sub-optimal results found in \cite{b2-topunov, Khlebnikov2011}.
\end{itemize}

The results presented in this paper extend the novel approach proposed in \cite{peregudin2023new} to the case of discrete-time systems.

\textbf{Notation.} 
 $\mathbb{Z}_{\ge 0}$ denote the non-negative integer numbers.  
 For $n$-dimensional signals $g: \mathbb{Z}_{\ge 0} \rightarrow \mathbb{R}^n$  and $p \ge 1$ we use the norms $\left\| g \right\|_p \doteq \left( \sum_{k=0}^{\infty} \left|g_{k} \right|^p   \right)^{1/p}$ and $\left\| g \right\|_\infty \! \doteq \max_{k=0}^{\infty} \left|g_{k} \right|$,
where $|g_{k}|^2 \doteq g^\top_{k} g^{\vphantom{\top}}_{k} $. The spectral radius of a square matrix $A$ is denoted as $\rho \! \left( A \right) \doteq \max_i \left| \lambda_i \! \left( A \right) \right|$, where $\lambda_i \! \left( A \right)$ is the $i$-th eigenvalue of the matrix. Matrix $A$ is said to be stable if $\rho \! \left( A \right) < 1$. The parallel configuration of  systems $\mathrm{S}_1$ and $\mathrm{S}_2$ is denoted as a sum $\mathrm{S}_1 + \mathrm{S}_2$. The series configuration of systems $\mathrm{S}_1$ and $\mathrm{S}_2$ is denoted as a right-to-left product $\mathrm{S}_2 \mathrm{S}_1$.

\section{The \texorpdfstring{$\varepsilon$}{}-norm for Discrete Systems}
Consider a discrete linear time-invariant system
\begin{equation}
    \mathrm{S}: \,
    \begin{cases}
        x_{k+1}=Ax_k+Bu_k, \\ 
        y_k = Cx_k,
    \end{cases}
    \label{DSOE}
\end{equation}
where $x_k \in \mathbb{R}^n$ is the state, $u_k \in \mathbb{R}^m$ is the control input, $y_k \in \mathbb{R}^p$ is the output, and $A$, $B$, $C$ are real matrices of corresponding sizes. Assume that $(A,B)$ is controllable, $(C,A)$ is observable, and $A$ is stable.

Define the reachable set as
\begin{equation*}
    \mathcal{R}_\infty \doteq \left\{ x_k \in \mathbb{R}^n \mid \mathrm{S}, \; x_0=0, \; \left\| u \right\|_\infty \leq 1  \right\}.
\end{equation*}
This set encompasses all states of the system \eqref{DSOE} that can be attained under arbitrary unit-bounded inputs. 
Ellipsoidal approximations of $\mathcal{R}_\infty$ are the subject of the invariant ellipsoids method \cite{b1-nazin, Khlebnikov2011}. The following theorem is known.

\begin{theorem} [\hspace{1sp}\cite{b1-nazin, Khlebnikov2011}]
\label{thm-P}
    If $\alpha \in \left(\rho^2 \! \left( A \right), 1 \right)$ and $\Pa \succ 0$ is a solution of the discrete Lyapunov equation
    \begin{equation}
        \frac{1}{\alpha}A\Pa A^\top -\Pa +\frac{1}{1-\alpha}BB^\top=0,
        \label{DLEEP}
    \end{equation}
    then $\mathcal{R}_\infty \subset \left\{ x \mid x^\top \Pa^{-1} x \leq 1 \right\}$.
\end{theorem}

We provide an alternative proof of Theorem \ref{thm-P} based on the Cauchy-Schwartz inequality. For the detailed proofs of this and all subsequent theorems, please refer to the \hyperlink{appendix}{Appendix}. 

Define the hardly observable set as
 \begin{equation*}
    \mathcal{O}_1 \doteq \left\{ x_0 \in \mathbb{R}^n \mid \mathrm{S}, \; u_k=0, \; \left\| y \right\|_1 \leq 1  \right\}.
\end{equation*}
This set comprises all initial states of  \eqref{DSOE} that lead to an output unit-bounded in the $1$-norm. 
The following theorem provides an ellipsoidal approximation of $\mathcal{O}_1$, which is a new result.
\begin{theorem}
    If $\alpha \in \left(\rho^2 \! \left( A \right), 1 \right)$ and $\Qa \succ 0$ is a solution of the discrete Lyapunov equation
    \begin{equation}
        \frac{1}{\alpha}A^\top \Qa A -\Qa +\frac{1}{1-\alpha}C^\top C=0,
        \label{DLEEQ}
    \end{equation}
    then $ \left\{ x \mid x^\top \Qa x \leq 1 \right\} \subset \mathcal{O}_1 $.
\end{theorem}

Surprisingly, the set $\mathcal{R}_\infty$ is dual to the set $\mathcal{O}_1$ and not to $\mathcal{O}_\infty$, which would have a similar definition but with the bound $\|y\|_\infty \le 1$. Moreover, the ellipsoidal approximations of reachable sets are external, while nearly unobservable sets are approximated from the inside. Continuous systems' primary investigation of these dual relations can be found in \cite{peregudin2023new}.

The main focus of this paper is the minimization of the reachable set of a closed-loop system, thereby ensuring the system's robust attenuation of a large class of bounded disturbances. To achieve this goal, we introduce the concept of the $\varepsilon(\alpha)$-norm of the system \eqref{DSOE} that quantifies the size of the ellipsoidal approximation of the reachable set. To be precise, its square equals to the sum of the squares of the semi-axes of the ellipsoid $\left\{ y \mid y^\top ( C \Pa C^\top )^{-1}y \le 1 \right\}$, which is the image of $\left\{ x \mid x^\top \Pa^{-1} x \le 1 \right\}$ under the linear mapping $x \mapsto y = Cx$. Notably, the value of the $\varepsilon(\alpha)$-norm is also associated with the size of the inner ellipsoidal approximation of the hardly observable set, as established by the following theorem.


\begin{theorem}
    If $\alpha \in \left(\rho^2 \! \left( A \right), 1 \right)$ and $\Pa$, $\Qa \succ 0$ are the solutions of \eqref{DLEEP} and \eqref{DLEEQ}, then 
    \begin{equation*}
        \left\| \mathrm{S} \right\|^2_{\varepsilon(\alpha)} \doteq \tr \left( C\Pa C^\top \right)= \tr \left( B^\top \Qa B \right).
        \label{EAND}
    \end{equation*}
        \vspace{-\baselineskip}
\end{theorem}
    With this duality, we also  introduce a parameter-independent $\varepsilon$-norm for the system \eqref{DSOE}. This norm corresponds to the tightest ellipsoidal approximations achievable by \eqref{DLEEP}-\eqref{DLEEQ} and is defined as the minimal 	attainable $\varepsilon(\alpha)$-norm:
    \begin{equation*}
        \left\| \mathrm{S} \right\|_{\varepsilon}\doteq \min_\alpha{\left\| \mathrm{S} \right\|}_{\varepsilon(\alpha)}.
        \label{END}
    \end{equation*}
In the next section, we will present results on synthesis that are optimal with respect to the $\varepsilon$-norm of the closed-loop system.

\section{Main Results}
\subsection{State-Feedback. \texorpdfstring{$\varepsilon$}{}-Optimal Controller}

Consider a discrete linear time-invariant closed-loop system
\begin{equation*}
    \mathrm{S}_K: \,
    \begin{cases}
        x_{k+1}=Ax_k+B u_k+B_w w_k, \\
        z_k=Cx_k+Du_k, \\
        u_k=Kx_k.
    \end{cases}
\end{equation*}
where $x_k \in \mathbb{R}^n$ is the state, $u_k \in \mathbb{R}^m$ is the control input, $w_k \in \mathbb{R}^{\bar{m}}$ is the external disturbance such that $\left\| w \right\|_\infty \leq 1$, $z_k \in \mathbb{R}^p$ is the regulated output, $A$, $B$, $B_w$, $C$, $D$, $K$ are real matrices of corresponding sizes, and $K$ is the feedback gain to be designed. Assume that $\left(A, B \right)$ is stabilizable, $\left( C,  A \right)$ is observable, and $C^\top D =0$.

The goal is to determine the optimal gain $K$ that minimizes the $\varepsilon$-norm of the system, which can be expressed as
\begin{equation*}
    \left\| \mathrm{S}_K \right\|_\varepsilon \; \longrightarrow \; \min .
\end{equation*}
The following theorem gives the optimal feedback gain $K$ that minimizes the $\varepsilon(\alpha)$-norm of the system.
\begin{theorem}
    \label{thm-K}
    Let $\Qa \succ 0$ be the solution of the discrete Riccati equation
    \begin{multline}
    \frac{1}{\alpha}A^\top \Qa B \left( B^\top \Qa B +\frac{\alpha}{1-\alpha}D^\top D  \right)^{-1} B^\top \Qa A \\
    -\frac{1}{\alpha}A^\top \Qa A + \Qa -\frac{1}{1-\alpha}C^\top C = 0
    \label{DAREC}
    \end{multline}
    for some $0 < \alpha < 1$. If the feedback gain $K$ be given by
    \begin{equation}
        K=-\left( B^\top \Qa B +\frac{\alpha}{1-\alpha}D^\top D  \right)^{-1} B^\top \Qa A,
        \label{DAREK}
    \end{equation}
    then the $\varepsilon(\alpha)$-norm of the closed-loop system attains the minimum possible value of
    \begin{equation*}
        \left\| \mathrm{S}_K \right\|_{\varepsilon(\alpha)} = \sqrt{\tr \left( B^\top_w \Qa B_w \right)}.
    \end{equation*}
    \vspace{-\baselineskip}
\end{theorem}

To find the optimal controller with respect to the $\varepsilon$-norm, the parameter $0 < \alpha < 1$ can be tuned iteratively to find the minimum of the function $\alpha \mapsto \tr \left( B^\top_w \Qa B_w \right)$. Numerical algorithms can be used to accelerate this process.

\begin{remark}
    Theorem \ref{thm-K} establishes an exact solution to the optimal control problem $ \left\| \mathrm{S}_K \right\|_\varepsilon \to \min$.  While the same matrices $\Qa$ and $K$ can be obtained using the optimization procedure proposed in \cite{b1-nazin}, this method requires solving a system of LMIs with $\frac{1}{2}n(n+1)+\frac{1}{2}m(m+1)+mn$ variables. In contrast, with the proposed approach one only has to solve a single Riccati equation with $\frac{1}{2}n(n+1)$ variables. Therefore, the proposed approach offers greater computational efficiency, particularly for high-order systems or those with numerous inputs. Even for second order examples (using standard \texttt{cvx} software and MATLAB function \texttt{are}), Riccati equation \eqref{DAREC} can be solved approximately 5-10 times faster than the system of LMIs presented \cite{b1-nazin}. The difference becomes particularly evident when iterating the parameter $\alpha$.
\end{remark}

\subsection{Filtering. \texorpdfstring{$\varepsilon$}{}-Optimal Observer}
Consider a discrete linear time-invariant plant
\begin{equation*}
    \mathrm{S}_L : \,
    \begin{cases}
        x_{k+1}=Ax_k+B w_k, \\
        y_k=Cx_k+Dw_k.
    \end{cases}
\end{equation*}
where $x_k \in \mathbb{R}^n$ is the state, $w_k \in \mathbb{R}^{\bar{m}}$ is the disturbance such that $\left\| w \right\|_\infty \leq 1$, $y_k \in \mathbb{R}^p$ is the measured output, $A$, $B$, $C$, $D$ are real matrices of corresponding sizes. Assume that $\left(C,  A\right)$ is detectable, $\left( A,  B\right)$ is controllable, and $B D^\top =0$.

Along with the plant, consider a linear observer
\begin{equation*}
    \mathrm{\hat{S}}_L : \,
    \begin{cases}
        \hat{x}_{k+1}=A\hat{x}_k+L\left(\hat{y}_k-y_k \right), \\
        \hat{y}_k=C\hat{x}_k,\\
        z_k=C_z \left(x_k - \hat{x}_k \right),
    \end{cases}
\end{equation*}
where $\hat{x}_k \in \mathbb{R}^n$ is the estimated state, $\hat{y}_k \in \mathbb{R}^p$ is the estimated output, $z_k \in \mathbb{R}^{\bar{p}}$ is the observer error, $C_z $, $L$ are real matrices of corresponding sizes, and $L$ is the observer gain to be designed.
The closed-loop system can be expressed as the product  $\hat{\mathrm{S}}_L \mathrm{S}_L$, in which the input is the  disturbance $w$ and the output is the observer error $z$.

The goal is to determine the optimal gain $L$ that minimizes the $\varepsilon$-norm of the system, which can be expressed as
\begin{equation*}
    \big\| \, \hat{\mathrm{S}}_L \mathrm{S}_L \big\|_\varepsilon \; \longrightarrow \; \min .
\end{equation*}

The following theorem gives the optimal observer gain $L$ that minimizes the $\varepsilon(\alpha)$-norm of the system.
\begin{theorem}
\label{thm-L}
    Let $\Pa \succ 0$ be the solution of the discrete Riccati equation
    \begin{multline}
        \frac{1}{\alpha}A \Pa C^\top \left( C \Pa C^\top +\frac{\alpha}{1-\alpha}DD^\top \right)^{-1} C \Pa A^\top \\- \frac{1}{\alpha}A \Pa A^\top + \Pa -\frac{1}{1-\alpha}B B^\top = 0
        \label{DAREO}
    \end{multline}
    for some $0 < \alpha < 1$. If the observer gain $L$ be given by
    \begin{equation}
        L=-A \Pa C^\top \left( C \Pa C^\top +\frac{\alpha}{1-\alpha}DD^\top   \right)^{-1},
        \label{DAREL}
    \end{equation}
    then the $\varepsilon(\alpha)$-norm of the closed-loop system attains the minimum possible value of
    \begin{equation*}
        \big\| \, \hat{\mathrm{S}}_L \mathrm{S}_L \big\|_{\varepsilon(\alpha)} = \sqrt{\tr \left( C_z \Pa C_z^\top \right)}.
    \end{equation*}
    \vspace{-\baselineskip}
\end{theorem}

To find the optimal filter with respect to the $\varepsilon$-norm, the parameter $0 < \alpha < 1$ can be tuned iteratively to find the
minimum of the function $ \alpha \mapsto \tr \left( C_z \Pa C_z^\top \right)$. Numerical algorithms can be used to accelerate this process.

\begin{remark}
    Theorem \ref{thm-L} establishes an exact solution to the optimal control problem $\| \hat{\mathrm{S}}_L \mathrm{S}_L \|_\varepsilon \to \min$.  While the same matrices $\Pa$ and $K$ can be obtained using the optimization procedure proposed in \cite{Khlebnikov2011}, this method requires solving a system of LMIs with $n(n+1)+np$ variables. In contrast, with the proposed approach one only has to solve a single Riccati equation with $\frac{1}{2}n(n+1)$ variables. Therefore, the proposed approach offers greater computational efficiency, particularly for high-order systems or those with numerous measurable outputs. The difference becomes particularly evident when iterating the parameter $\alpha$.
\end{remark}

\subsection{Output-Feedback. \texorpdfstring{$\varepsilon$}{}-Optimal Controller}
\label{section2c}
Consider a discrete linear time-invariant plant
\begin{equation}
    \begin{cases}
        x_{k+1}=Ax_k+B_1 w_k +B_2 u_k, \\
        y_k = C_1 x_k + D_1 w_k, \\
        z_k = C_2 x_k + D_2 u_k.
    \end{cases}
    \label{FDSOE}
\end{equation}
where $x_k \in \mathbb{R}^n$ is the state, $u_k \in \mathbb{R}^m$ is the control input, $w_k \in \mathbb{R}^{\bar{m}}$ is the external disturbance such that $\left\| w \right\|_\infty \leq 1$, $y_k \in \mathbb{R}^p$ is the measured output, $z_k \in \mathbb{R}^p$ is the regulated output, $A$, $B_i$, $C_i$, $D_i$ are real matrices of corresponding sizes. Assume that $\left(A, \; B_2 \right)$ is stabilizable, $\left( C_2,  A \right)$ is observable, $\left( C_1,  A \right)$ is detectable, $\left( A,  B_1 \right)$ is controllable, $B_1D_1^\top =0$, and $C_2^\top D_2 =0$.

Along with the plant, consider a linear controller
\begin{equation}
    \begin{cases}
        \hat{x}_{k+1}=A\hat{x}_k+B_2 u_k+L\left(\hat{y}_k-y_k \right), \\
        \hat{y}_k=C_1\hat{x}_k, \\
        u_k=K \hat{x}_k.
    \end{cases} 
    \label{CDSOE}
\end{equation}
where $\hat{x}_k \in \mathbb{R}^n$ is the estimated state, $\hat{y}_k \in \mathbb{R}^p$ is the estimated output, $K$ and $L$ are controller gains to be determined. Denote the closed-loop system \eqref{FDSOE}-\eqref{CDSOE} as $\mathrm{S}_{KL}$, in which the input is the disturbance $w$ and the output is the regulated output $z$.

The goal is to find the optimal gains $K$ and $L$ that minimize the $\varepsilon$-norm of the system, which can be expressed as
\begin{equation*}
    \left\| \mathrm{S}_{KL} \right\|_\varepsilon \; \longrightarrow \; \min .
\end{equation*}

The following theorem establishes the optimal controller gains $K$ and $L$ that minimize the $\varepsilon(\alpha)$-norm of the system.
\begin{theorem}
\label{thm-KL}
    Let $\Pa, \Qa \succ 0$ be the solution of the discrete Riccati equations
    \begin{multline*}
        \frac{1}{\alpha}A \Pa C_1^\top \left( C_1 \Pa C_1^\top +\frac{\alpha}{1-\alpha}D_1D_1^\top \right)^{-1} C_1 \Pa A^\top \\ - \frac{1}{\alpha}A \Pa A^\top + \Pa -\frac{1}{1-\alpha}B_1 B_1^\top =0,
    \end{multline*}
    \vspace{-\baselineskip}
    \begin{multline*}
        \frac{1}{\alpha}A^\top \Qa B_2 \left( B_2^\top \Qa B_2 +\frac{\alpha}{1-\alpha}D_2^\top D_2  \right)^{-1} B_2^\top \Qa A \\ -\frac{1}{\alpha}A^\top \Qa A + \Qa -\frac{1}{1-\alpha}C_2^\top C_2 =0,
    \end{multline*}
    for some $0 < \alpha < 1$. If the controller gains be given by
     \begin{align}
        & K=-\left( B_2^\top \Qa B_2 +\frac{\alpha}{1-\alpha}D_2^\top D_2  \right)^{-1} B_2^\top \Qa A, \label{thmKL1}\\
        & L=-A \Pa C_1^\top \left( C_1 \Pa C_1^\top +\frac{\alpha}{1-\alpha}D_1D_1^\top   \right)^{-1}, \label{thmKL2}
    \end{align}
    then the $\varepsilon(\alpha)$-norm of the closed-loop system attains the minimum possible value of
    \begin{equation*}
        \left\| \mathrm{S}_{KL} \right\|_{\varepsilon(\alpha)} = \sqrt{\tr \left( B^\top_1 \Qa B_1 \right) + \tr \left( R K\Pa K^\top R^\top \right)}
    \end{equation*}
    where $R$ is such that $R^\top R =  \frac{1-\alpha}{\alpha} B_2^\top \Qa B_2 +D_2^\top D_2$.
\end{theorem}

Note that finding the value of $R$ may be considered unnecessary, given the cyclic property of the trace. To find the optimal output-feedback controller with respect to the $\varepsilon$-norm, the parameter $0 < \alpha < 1$ can be tuned iteratively to minimize the function $\alpha \mapsto \tr \left( B^\top_1 \Qa B_1 \right) + \tr \left(  K\Pa K^\top R^\top R \right)$.

\begin{remark}
    Theorem \ref{thm-KL} establishes a new result that was previously unknown. In the subsequent section we will compare this result to existing ones. 
\end{remark}

\section{Comparison with known results}

We compare the newly established optimal output-feedback control result with the previously known sub-optimal methods \cite{b2-topunov} and \cite{Khlebnikov2011} through the following example.

Consider a discrete time-invariant plant \eqref{FDSOE} with matrices
\begin{alignat*}{9}
    &A &&= &&\begin{bmatrix} 0 & 1\\ -1 & 0 \end{bmatrix}, 
    &&B_1 &&= &&\begin{bmatrix} 8 & 0\\ 8 & 0 \end{bmatrix}, \quad
    &&B_2 &&= &&\begin{bmatrix} 4\\ 8 \end{bmatrix}, \\[0.2em]
    &C_1 &&= &&\begin{bmatrix} 2 & 0\\ -6 & -2 \end{bmatrix}, \quad 
    &&D_1 &&= &&\begin{bmatrix} 0 & 0\\ 0 & 2 \end{bmatrix}, && && &&\\[0.2em]
    &C_2 &&= &&\begin{bmatrix} 8 & 6\\ 6 & -4\\ 0 & 0 \end{bmatrix}, 
    &&D_2 &&= &&\begin{bmatrix} 0\\ 0\\ 4 \end{bmatrix}.  && && && 
\end{alignat*}
We apply the methods \cite{b2-topunov}, \cite{Khlebnikov2011} (which yield the same results) and the new method outlined in Section \ref{section2c}. The results are presented in Figures \ref{fig1}, \ref{fig2}, and Table \ref{tablo}.

\begin{figure}[!t]
    \centerline{\includegraphics[width=1.03\columnwidth]{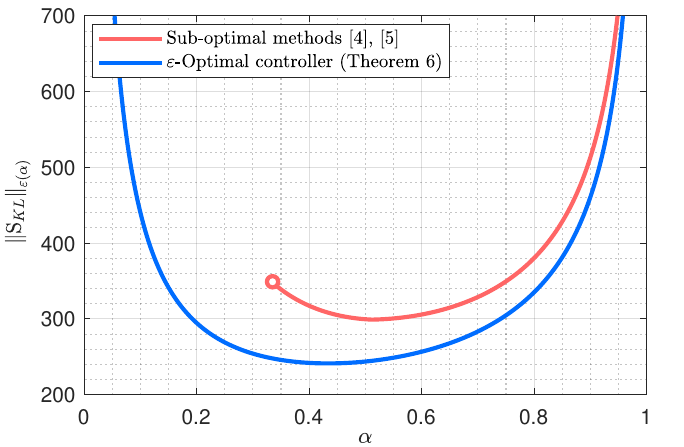}}
    \caption{A comparison of of the closed-loop system's $\varepsilon(\alpha)$-norm between the sub-optimal methods \cite{b2-topunov}, \cite{Khlebnikov2011} and the optimal controller proposed in Theorem \ref{thm-KL}. }
    \label{fig1}
\end{figure}

Figure \ref{fig1} illustrates the $\varepsilon(\alpha)$-norm as a function of $\alpha$ in the interval $(0, 1)$ for the two algorithms. Note that the sub-optimal methods \cite{b2-topunov}, \cite{Khlebnikov2011} do not yield solutions for small values of $\alpha$ due to degeneracy in the corresponding LMIs. In contrast, the proposed approach provides a solution for all values of $\alpha \in (0, 1)$, and the corresponding $\varepsilon(\alpha)$-norms are consistently smaller than those for the previously known methods.

\begin{table}[b]
    \vspace{-\baselineskip}
    \centering
    \caption{Comparison of the obtained output-feedback controllers}
    \label{tablo}
    \begin{tabular}{cc}
    \toprule
     \quad Sub-optimal methods \cite{b2-topunov}, \cite{Khlebnikov2011} \quad & \quad $\varepsilon$-Optimal controller (Theorem \ref{thm-KL}) \quad \\
    \midrule 
     $K = \begin{bmatrix} 0.0919 & -0.0661 \end{bmatrix}$ & $K = \begin{bmatrix}  0.0928 & -0.0643 \end{bmatrix}$ \\[0.3em]
     $L = \begin{bmatrix} 1.32 & 0.407 \\  0.681 & 0.0927 \end{bmatrix}$ & $L = \begin{bmatrix}  0.626 & 0.281 \\ 0.5 & 0 \end{bmatrix}$ \\
      \midrule 
   $\|\mathrm{S}_{KL}\|_\varepsilon=299.0$ & $\|\mathrm{S}_{KL}\|_\varepsilon=241.2$ \\
       \bottomrule
    \end{tabular}
    \label{table}
    \vspace{-\baselineskip}
\end{table}


In Figure \ref{fig2}, we visualize the invariant ellipsoids and discrete state trajectories for both methods.  The proposed approach yields a smaller ellipsoid, signifying improved attenuation of bounded disturbances. Note that the sizes of these ellipsoids are quantified by the values of the $\varepsilon$-norm.

Table \ref{tablo} displays the controller parameters and the corresponding $\varepsilon$-norm values obtained with both methods. It is evident that the new approach yields a smaller value of the $\varepsilon$-norm, demonstrating its superiority over the previously known methods. Moreover, according to Theorem \ref{thm-KL}, our result is proven to be optimal and cannot be improved.

\begin{figure}[!t]
    \centerline{\includegraphics[width=1.03\columnwidth]{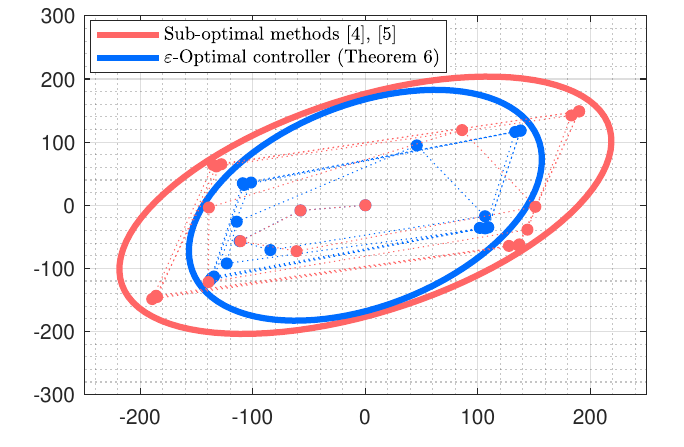}}
    \caption{A comparison of the closed-loop system's invariant ellipsoids (in 2D-projection) and sample state trajectories obtained with sub-optimal methods \cite{b2-topunov}, \cite{Khlebnikov2011} and the optimal method proposed in Theorem \ref{thm-KL}.}
    \label{fig2}
\end{figure}

\section{Conclusion}

This paper has established a duality relation between ellipsoidal approximations of reachable and hardly observable sets for discrete-time systems. By leveraging this duality, we have introduced a novel approach to addressing optimal state-feedback and filtering problems in the presence of bounded disturbances. Our primary contribution lies in providing a new solution to the output-feedback control problem for discrete-time systems subjected to bounded disturbances. This solution is not only proven optimal with respect to the $\varepsilon$-norm but also clearly demonstrated to outperform prior sub-optimal results.

\appendix

\hypertarget{appendix}{}
\label{appendix}

\subsection{Proof of Theorem 1}
\begin{proof}
For $\alpha \in (\rho^2 \! \left( A \right), 1 )$, the solution of equation \eqref{DLEEP} has the form $\Pa = \sum_{i=0}^\infty \frac{\alpha^{-i}}{\left(1-\alpha \right) }  A^i BB^\top \! \left(A^\top \right)^i$.
Let $c \in \mathbb{R}^{1 \times n}$ be an arbitrary nonzero row vector. Then for all $ k\in \mathbb{Z}_{\ge 0}$, 
\begin{equation*}
    \begin{aligned}
    cx_{k+1}  x^\top_{k+1}c^\top & = \left( \sum_{i=0}^k cA^{k-i}B u_i \right)^2 \\ &\le \left( \sum_{i=0}^k \alpha^{i-k} \left( cA^{k-i}B u_i \right)^2 \right) \left( \sum_{i=0}^k \alpha^{k-i} \right) \\  
    &\le \left( \sum_{i=0}^k \alpha^{i-k} \left( cA^{k-i}B u_i \right)^2 \right) \left( \sum_{i=-\infty}^k \alpha^{k-i} \right) \\
\end{aligned}
\end{equation*}
\begin{equation*}
    \begin{aligned}
    \phantom{cx_{k+1}  x^\top_{k+1}c^\top} 
    &= \sum_{i=0}^k \frac{\alpha^{i-k}}{\left(1-\alpha \right) }  \left( cA^{k-i}B u_i \right)^2 \\ 
    &\le \sum_{i=0}^k \frac{\alpha^{i-k}}{\left(1-\alpha \right) }  \left| cA^{k-i}B \right|^2 \left| u_i \right|^2 \\ 
    &\le \max_{i=0}^k \left| u_i \right|^2 \sum_{i=0}^k \frac{\alpha^{-i}}{\left(1-\alpha \right) }  cA^i BB^\top \! \left(A^\top \right)^i c^\top \\ 
    & \le \left\| u \right\|^2_\infty c \Pa c^\top,
\end{aligned}
\end{equation*}
where the first inequality holds by the Cauchy-Schwarz theorem. It can also be noted that since $k$ is arbitrary, we can substitute $x_k$ for $x_{k+1}$ without loss of generality.
Given that $c$ is nonzero vector, the inequality is equivalent to the LMI $x_k x^\top_k \preceq \left\| u \right\|^2_\infty \Pa$.
By applying the Schur complement property, we obtain $x^\top_k \Pa^{-1} x_k \le \left\| u \right\|^2_\infty$.
It follows that
\begin{equation*}
    \left\| u \right\|^2_\infty \le 1 \quad \Longrightarrow \quad x^\top_k \Pa^{-1} x_k \le 1.
\end{equation*} 
This completes the proof.
\end{proof}

\subsection{Proof of Theorem 2}
\begin{proof}
For $\alpha \in \left(\rho^2 \! \left( A \right), 1 \right)$, the solution of equation \eqref{DLEEQ} has the form $\Qa = \sum_{i=0}^\infty \frac{\alpha^{-i}}{\left( 1- \alpha \right) } \left( A^\top \right)^i C^\top C A^i$.
Let $x_0 \in \mathbb{R}^n$ be an arbitrary initial state. Then we have
\begin{align*}
    x_0^\top \Qa x_0 
    &= \sum_{i=0}^\infty \frac{\alpha^{-i}}{\left( 1- \alpha \right) } x_0^\top \left( A^\top \right)^i C^\top C A^i x_0 \\ 
    &= \sum_{i=0}^\infty \frac{\alpha^{-i}}{\left( 1- \alpha \right) } \left|C A^i x_0 \right|^2 \\ 
    &= \left( \sum_{i=0}^\infty \alpha^{-i} \left|C A^i x_0 \right|^2 \right) \left( \sum_{i=0}^\infty \alpha^i \right) \\ 
    &\ge \left( \sum_{i=0}^\infty  \left|C A^i x_0 \right| \right)^2 = \left\| y \right\|^2_1,
\end{align*}
where the inequality holds by the Cauchy-Schwarz theorem. It follows that
\begin{equation*}
    x_0^\top \Qa x_0 \le 1 \quad \Longrightarrow \quad \left\| y \right\|^2_1 \le 1.
\end{equation*}
This completes the proof.
\end{proof}
\subsection{Proof of Theorem 3}
\begin{proof}
Given the solutions of the equations \eqref{DLEEP} and \eqref{DLEEQ} it is straightforward to discern that
\begin{align*}
    \trace \left( C \Pa C^\top \right)  & = \trace \Big( \sum_{i=0}^\infty \frac{\alpha^{-i}}{\left(1-\alpha \right) } C A^i BB^\top \! (A^\top )^i C^\top \Big) 
    \\ &= \trace \Big( \sum_{i=0}^\infty \frac{\alpha^{-i}}{\left(1-\alpha \right) } B^\top \! (A^\top )^i C^\top C A^i B \Big)  
    \\ &= \trace \left( B^\top \Qa B \right) \vphantom{\Big(},
\end{align*}
establishing the desired duality.
\end{proof}

\subsection{Proof of Theorem 4}

\begin{proof}
Consider the closed-loop system
\begin{equation*}
    \mathrm{S}_K: \,
    \begin{cases}
        x_{k+1}=\left( A+BK \right)x_k+B_w w_k, \\ 
        z_k = \left( C+DK \right)x_k
    \end{cases}
\end{equation*}
and apply equation \eqref{DLEEQ} to obtain 
\begin{multline*}
    \frac{1}{\alpha} \left( A+BK \right)^\top \Qa \left( A+BK \right) - \Qa \\+ \frac{1}{1-\alpha} \left( C+DK \right)^\top \left( C+DK \right) = 0.
\end{multline*}
Expanding and taking into account $C^\top D=0$, we get
\begin{multline*}
    \frac{1}{\alpha}A^\top \Qa A - \Qa + \frac{1}{1-\alpha}C^\top C + \frac{1}{1-\alpha}K^\top D^\top D K 
    \\+ \frac{1}{\alpha} A^\top \Qa BK + \frac{1}{\alpha}K^\top B^\top \Qa A + \frac{1}{\alpha} K^\top B^\top \Qa BK = 0. 
\end{multline*}
Denoting $\Phi = \big( B^\top \Qa B +\frac{\alpha}{1-\alpha}D^\top D  \big)^{-1}$ and completing the squares yield 
\begin{multline*}
    \frac{1}{\alpha} \left(K + \Phi B^\top \Qa A \right)^\top \Phi^{-1}  \left(K + \Phi B^\top \Qa A \right) + \frac{1}{1-\alpha}C^\top C \\ + \frac{1}{\alpha}A^\top \Qa A - \Qa - \frac{1}{\alpha}A^\top \Qa B \Phi B^\top \Qa A  =0.
\end{multline*}
Observe that the first term of this equation is always positive-semidefinite. By the comparison theorem \cite{RAN198863}, the minimal possible $\Qa \succ 0$ for a given $\alpha$ is achieved if the first term minimizes to zero, which leads to \eqref{DAREK}. This simplifies the above expression into the discrete Riccati equation form \eqref{DAREC}.
With the assumptions on controllability and observability, this equation possesses a unique positive definite solution $\Qa $ that also satisfies the Lyapunov equation \eqref{DLEEQ}, thereby implying asymptotic stability of the closed-loop system. By definition of the $\varepsilon(\alpha)$-norm we have $\left\| \mathrm{S}_K \right\|_{\varepsilon(\alpha)}^2 = \tr \left( B^\top_w \Qa B_w \right)$.
\end{proof}
\subsection{Proof of Theorem 5}
\begin{proof}
Consider the discrete closed-loop system
\begin{equation*}
    \mathrm{S}_L: \,
    \begin{cases}
        e_{k+1}=\left( A+LC \right)e_k+ \left(B + LD \right) w_k, \\ 
        z_k = C_z e_k,
    \end{cases}
\end{equation*}
where $e_k=x_k - \hat{x}_k$, and apply equation \eqref{DLEEP} to obtain
\begin{multline*}
    \frac{1}{\alpha} \left( A +LC \right) \Pa \left( A+LC \right)^\top - \Pa \\+ \frac{1}{1-\alpha} \left( B+LD \right) \left( B+LD \right)^\top  = 0.
\end{multline*}
Given the dual nature of this theorem's proof to the one of Theorem \ref{thm-K}, it follows that the minimum possible $\Pa \succ 0$ for given $0<\alpha<1$ is reached if $L$ is selected as \eqref{DAREL}. This simplifies above expression into the discrete Riccati equation form \eqref{DAREO}, which has a unique stabilizing solution by the assumptions on controllability and observability. By definition of the $\varepsilon(\alpha)$-norm we have $\left\| \mathrm{S}_L \right\|_{\varepsilon(\alpha)}^2 = \tr \left( C_z \Pa C_z^\top \right)$.
\end{proof}

\subsection{Lemma 1}
\begin{lemma}
\label{LEMM}
    Let systems $\mathrm{S}_1$ and $\mathrm{S}_2$ be given as
    \begin{align*}
        \mathrm{S}_1&: \,
        \begin{cases}
            x_{k+1}= \bar{A} x_k + \bar{B} w_k, \\
            h_k= \bar{C} x_k
        \end{cases} \\
        \mathrm{S}_2&: \,
        \begin{cases}
            x_{k+1}= \left( A+BK \right)x_k+ B h_k,  \\
            z_k= \left( C +DK\right) x_k +D h_k,
        \end{cases}
    \end{align*}  
    where $K$ is calculated from \eqref{DAREC}-\eqref{DAREK}. Then the $\varepsilon(\alpha)$-norm of the series configuration of these systems can be computed as
    \begin{equation*}
        \left\| \mathrm{S}_2 \mathrm{S}_1 \right\|^2_{\varepsilon (\alpha)} = \left\| \mathrm{S}' \right\|^2_{\varepsilon (\alpha)},  
    \end{equation*}
    where
    \begin{equation*}
        \mathrm{S}': \,
        \begin{cases}
            x_{k+1}= \bar{A} x_k + \bar{B} w_k,   \\
            z_k= R \bar{C} x_k, 
        \end{cases}
    \end{equation*}
    and $R$ is such that $R^\top R =  \frac{1-\alpha}{\alpha} B^\top \Qa B +D^\top D$.
\end{lemma}

\begin{proof}
The series connection $\mathrm{S}_2 \mathrm{S}_1$ can be expressed as 
\begin{align*}
    \left[ \begin{matrix} \bar{x}_{k+1} \\ x_{k+1} \end{matrix} \right] &= 
    \left[ \begin{matrix} \bar{A} & 0 \\B\bar{C} & A+BK \end{matrix} \right]
    \left[ \begin{matrix} \bar{x}_k \\ x_k \end{matrix} \right]  +
    \left[ \begin{matrix} \bar{B} \\ 0 \end{matrix} \right]  w_k, \\
    z_k &= \left[ \begin{matrix} D\bar{C} & C+DK\end{matrix} \right]
    \left[ \begin{matrix} \bar{x}_k \\ x_k \end{matrix} \right].
\end{align*}
It is straightforward to check, that
\begin{equation*}
    Q' = \left[ \begin{matrix} \bar{Q} & 0 \\ 0 & \Qa \end{matrix} \right]
\end{equation*}
is the solution of \eqref{DLEEQ}, where $\bar{Q}$ is calculated from 
\begin{equation*}
    \frac{1}{\alpha} \bar{A}^\top \bar{Q} \bar{A} - \bar{Q} +\frac{1}{1-\alpha} \bar{C}^\top R^\top R \bar{C} = 0,
\end{equation*}
where  $R$ is such that $R^\top R =  \frac{1-\alpha}{\alpha} B^\top \Qa B +D^\top D$.
Therefore, $\left\| \mathrm{S}_2 \mathrm{S}_1 \right\|^2_{\varepsilon (\alpha)} = \tr \left( \bar{B}^\top \bar{Q} \bar{B} \right)  = \left\| \mathrm{S}' \right\|^2_{\varepsilon (\alpha)}$, as desired.
\end{proof}

\subsection{Proof of Theorem 6}
\begin{proof}
Consider the following representation of the closed-loop system $\mathrm{S}_{KL}$:
\begin{align*}
    \left[ \begin{matrix} x_{k+1} \\ e_{k+1} \end{matrix} \right] &= 
    \left[ \begin{matrix} A+B_2K & -B_2K \\0 & A+LC_1 \end{matrix} \right]
    \left[ \begin{matrix} x_k \\ e_k \end{matrix} \right]  +
    \left[ \begin{matrix} B_1 \\ B_1+LD_1 \end{matrix} \right]  w_k, \\
    z_k &= \left[ \begin{matrix} \; C_2+D_2K & -D_2K \;\end{matrix} \right]
    \left[ \begin{matrix} x_k \\ e_k \end{matrix} \right].
\end{align*}
This system can be partitioned into three interconnected subsystems:
\begin{align*}
    \mathrm{S}_K&: \,
    \begin{cases}
        x_{k+1}= \left( A+B_2K \right)x_k+ B_1 w_k, \\
        z_k= \left( C_2 +D_2K\right) x_k,
    \end{cases} \\
    \mathrm{S}^*_L&: \,
    \begin{cases}
        e_{k+1}= \left( A+LC_1 \right)e_k+ \left( B_1+LD_1\right) w_k, \\
        h_k= -K e_k,
    \end{cases}\\
    \mathrm{S}^*_K&: \,
    \begin{cases}
        x_{k+1}= \left( A+B_2K \right)x_k+ B_2 h_k, \\
        z_k= \left( C_2 +D_2K\right) x_k +D_2 h_k.
    \end{cases}
\end{align*}
The vector $z$ of $\mathrm{S}_{KL}$ results from the parallel connection of $\mathrm{S}_K$ and the series connection of $\mathrm{S}^*_L$ and $\mathrm{S}^*_K$. Thus, it can be written as $z=\left(\mathrm{S}_{KL} \right)w =  \left(\mathrm{S}_K+ \mathrm{S}^*_K \mathrm{S}^*_L \right) w$.

By applying Lemma \ref{LEMM}, we conclude that
\begin{alignat*}{4}
    &\left\| \mathrm{S}_{KL} \right\|^2_{\varepsilon (\alpha)} &&= &&\left\| \mathrm{S}_{K} \right\|^2_{\varepsilon (\alpha)} + &&\left\| \mathrm{S}^*_K \mathrm{S}^*_L \right\|^2_{\varepsilon (\alpha)} \\
    & &&= &&\left\| \mathrm{S}_{K} \right\|^2_{\varepsilon (\alpha)} + &&\left\| \mathrm{S}'_L \right\|^2_{\varepsilon (\alpha)}, 
\end{alignat*}
where
\begin{equation*}
    \mathrm{S}'_L: \,
    \begin{cases}
        e_{k+1} = \left( A+LC_1 \right)e_k+ \left( B_1+LD_1\right) w_k, \\
        z_k= -R K e_k,
    \end{cases}
\end{equation*}
and $R$ is such that $R^\top R =  \frac{1-\alpha}{\alpha} B_2^\top \Qa B_2 +D_2^\top D_2$.

It is evident that system $\mathrm{S}_K$ corresponds to the system from Theorem \ref{thm-K}, while $\mathrm{S}'_L$ is equivalent to the one in Theorem \ref{thm-L}, and that the separation of norms is achieved. Therefore, the minimum value of the $\varepsilon(\alpha)$-norm for system $\mathrm{S}_{KL}$ is attained when the controller gains are given by \eqref{thmKL1}-\eqref{thmKL2}. It follows that the corresponding $\varepsilon(\alpha)$-norm can be expressed as
\begin{equation*}
    \left\| \mathrm{S}_{KL} \right\|^2_{\varepsilon(\alpha)} = \tr \left( B^\top_1 \Qa B_1 \right) + \tr \left( R K\Pa K^\top R^\top \right).
\end{equation*}
This completes the proof.
\end{proof}



\bibliographystyle{ieeetr}
\bibliography{file}


\end{document}